\documentclass[aps,pra,twocolumn,showkeys]{revtex4-1}
\usepackage{amsmath,paralist,amsthm,comment,amssymb}
\usepackage{times,color}
\usepackage{amscd}
\usepackage{epsfig}

\usepackage{subfigure}
\usepackage{algorithm} 
\usepackage{algorithmic}

\usepackage{tikz}
\usetikzlibrary{matrix,arrows}
\usetikzlibrary{shapes.geometric}
\usetikzlibrary{patterns} 
\usepackage{pgfplots}

\newcommand{\mc}[1]{\mathcal{#1}}
\newcommand{\nix}[1]{}

\newtheorem{theorem}{Theorem}
\newtheorem{corollary}[theorem]{Corollary}
\newtheorem{lemma}[theorem]{Lemma}

\begin{document}

\title{ Relation Between Surface Codes and Hypermap-Homology Quantum Codes}

\author{Pradeep Sarvepalli}
\email[]{pradeep@ee.iitm.ac.in}
\affiliation{Department of Electrical Engineering, Indian Institute of Technology Madras, Chennai 600 036, India}
\date{March 14, 2014} 

\begin{abstract}
Recently, a new class of quantum codes based on  hypermaps were proposed. These codes are obtained from embeddings of hypergraphs as opposed to surface codes which are obtained from the  embeddings of graphs. 
It is natural to compare these two classes of codes and their relation to each other. 
 In this context  two related questions are addressed in this paper: Can the parameters of hypermap-homology codes be superior to those of surface codes and  what is precisely the relation between these two classes of quantum codes? We show that a canonical hypermap code is identical to a surface code while a noncanonical hypermap code can be transformed to a surface code by CNOT gates alone. Our approach is  constructive; we construct the related surface code and the transformation involving CNOT gates.
\end{abstract}


\maketitle
\section{Introduction}

Surface codes,  proposed by Kitaev \cite{kitaev03}, are an extremely appealing class of codes for fault tolerant quantum computation \cite{dennis02,raussen07}. They have been generalized in various directions \cite{bravyi98,bullock07,bombin07b,bombin07,bombin06,tillich09, zemor09}. 
Recently,  a new generalization of surface codes was proposed in \cite{martin13}. In this approach quantum codes are constructed based on the homology of hypergraphs  rather than the homology of graphs, which is the case in surface codes. 
Hypergraphs are a generalization of graphs; the edges of a hypergraph can be incident on two or more vertices. Just as the embedding of graphs on surfaces gives rise to codes which are topological in nature \cite{kitaev03}, 
embeddings of hypergraphs also give rise to topological quantum codes. These codes have been termed hypermap-homology 
codes in \cite{martin13}; we also refer to them as hypermap codes.  

Arguably, the most popular surface code is Kitaev's toric code defined on a
square lattice of size $n\times n$ with periodic boundary conditions \cite{kitaev03}. This code is a $[[2n^2,2,n]]$ quantum code. On the other hand, Ref.~\cite{martin13} proposed a hypermap homology code defined on the  $n\times n $ square  lattice.  This is  a $[[3n^2/2,2,n]]$ code and 
it is more efficient than Kitaev's toric code with 
respect to the number of physical qubits required to protect the 
same number of logical qubits while maintaining  the same level of error correcting capability. 
This seemed to suggest that better quantum codes maybe obtained by using  hypergraphs. But there are other surface codes with better parameters than the $[[2n^2,2,n]]$ toric code. There exist surface codes whose parameters are $[[n^2+1,2,n]]$, see \cite{bombin07c,kovalev12}. 

This begs the question if hypermap codes improve upon the parameters of best surface codes?
A related question is the exact relation between hypermap codes and surface codes. 
It was also not apparent which codes could be realized using  hypermaps but not by embedding graphs on surfaces. 

In this paper we address these questions. The construction of hypermap codes
requires us to make a choice as to whether we will represent certain boundary maps using the standard basis 
or  not. We call those codes with standard basis as canonical hypermap codes and those with a nonstandard basis 
as  noncanonical hypermap codes. We show that for every canonical hypermap code there is a surface code which is 
identical to the hypermap code. This implies that the $[[3n^2/2,2,]]$ hypermap code, although it improves upon the $[[2n^2,2,n]]$ toric code, it is identical to another surface code.
We also show that every noncanonical hypermap code can be transformed to a surface 
code; we only need CNOT gates for this transformation.  In some cases, a 
noncanonical hypermap code can be identical to a surface code. A hypermap code that cannot be realized by a surface code must be a noncanonical code.

Our results imply that hypermap codes that improve upon surface codes or which cannot be realized as surface codes must be noncanonical hypermap codes.
Through our results, many questions related to hypermap codes can be posed as questions about 
surface codes and we can study  hypermap codes using surface codes. For instance, decoding of these codes can be studied in terms of related surface codes.  

This paper is structured as follows. In Section~\ref{sec:bac} after briefly reviewing surface codes, we give a  self-contained introduction to hypermap codes. Then in Section~\ref{sec:main} we present our main results which clarify the relation between canonical  and noncanonical hypermap codes on one hand
and the relation between hypermap codes and surface codes on the other. We then conclude with a brief discussion on the significance and consequences of our results.

\section{Background} \label{sec:bac}
\subsection{Surface codes}
We assume that the reader is familiar with the basics of quantum codes and stabilizer formalism \cite{gottesman97,calderbank98}. 
Let $\mathsf{G}$ be a graph with vertex set $\mathsf{V}(\mathsf{G})$ and edge set $\mathsf{E}(\mathsf{G})$; when there is no confusion 
we drop the dependence on $\mathsf{G}$. Consider an embedding of the graph on a surface $\Sigma$
i.e. a drawing of the graph on the surface so that no two edges cross each other. Denote the 
faces  of the embedding as $\mathsf{F}(\mathsf{G})$. We restrict our attention only to those embeddings in which the 
faces are homeomorphic to  open  discs. Such an embedding is also called a 2-cell embedding or a map. A 
(stabilizer) quantum code is obtained from the embedding as follows. On each edge we place a qubit. 
 We associate to each vertex an operator called vertex 
operator, denoted by  $A_v$ and to each face an operator, denoted $B_f$. The face operators are also 
sometimes referred to as plaquette operators. These operators are defined 
as follows:
\begin{eqnarray}
A_v &=& \prod_{u\in \delta(v)}X_u, \mbox{  where } \delta(v)=\{ u\mid (u,v) \in \mathsf{E}(\mathsf{G}) \}\label{eq:vertex-op}\\
B_f & =&\prod_{e\in \partial(f)} Z_e,\label{eq:face-op}
\end{eqnarray}
where $ \partial(f) =\{e\mid e \mbox{ is in the boundary of $f$} \}$; while $X$ and $Z$ are the Pauli matrices. 
The surface code is defined as the joint +1-eigenspace stabilized by the operators $A_v$ and $B_f$. 
In other words, it is the quantum code with the stabilizer 
\begin{eqnarray}
S&=& \langle A_v, B_f \mid v\in \mathsf{V}(\mathsf{G}), f\in \mathsf{F}(\mathsf{G})\rangle.
\label{eq:surf-stabilizer}
\end{eqnarray}
The stabilizer matrix of the surface code can be represented by 
$\left[\begin{array}{cc}\mathsf{I}_\mathsf{V}& 0 \\0 & \mathsf{I}_\mathsf{F} \end{array} \right]$, where 
$\mathsf{I}_\mathsf{V}$ is the vertex-edge incidence matrix of $\mathsf{G}$ and $\mathsf{I}_\mathsf{F}$ is the face-edge incidence matrix of $\mathsf{G}$. 
The surface code   is an $[[|\mathsf{E}|, 2g]]$ quantum code, where $g$ is the genus of the surface on which $\mathsf{G}$ is embedded.

\subsection{Hypermap homological codes}

We now review the hypermap-homological codes proposed in \cite{martin13}. The reader can find more details on 
these codes therein. 
Let $\Gamma$ be a hypergraph with vertex set $\mathsf{V}(\Gamma)$ and hyperedge set $
\mathsf{E}(\Gamma)$. A hyperedge is any nonempty subset of the vertex set.  If $\mathsf{E}(\Gamma)$ has only subsets of size two then $\Gamma$  reduces to a standard graph.  (We use Greek alphabet for hypergraphs only.) Hypergraphs are often studied by means of a bipartite graph representation. This bipartite graph is formed as follows:
Form a bipartite graph $\mathsf{G}_\Gamma$ with vertex set $\mathsf{V}(\mathsf{G}_\Gamma) = \mathsf{V}(\Gamma) \cup \mathsf{E}(\Gamma)$. 
Place an edge between $v \in  \mathsf{V}(\Gamma) $ and $e\in \mathsf{E}(\Gamma)$ if and only if  $e$ is incident on $v$.
We refer to the edges of the bipartite graph as darts or half-edges and denote the collection of 
darts by  $\mathsf{E}(\mathsf{G}_\Gamma)$. The darts will also be denoted as $\mathsf{W}(\Gamma)$.
 For any dart $i$, we denote the unique vertex in $\mathsf{V}(\Gamma)$ on which $i$ is incident by $v_{\ni i}$ and the unique hyperedge on which $i$
is incident by $e_{\ni i}$.

Consider the embedding of $\mathsf{G}_\Gamma$ on a surface $
\Sigma$. Denote the faces of  $\mathsf{G}_\Gamma$ by 
$\mathsf{F}(\mathsf{G}_\Gamma)$. An embedding of $\mathsf{G}_\Gamma$ is called 
a hypermap. The labeling of the darts in the hypermap is performed as follows: we place a label to the left of the dart as we move along the dart from a hyperedge to an adjacent vertex in $\mathsf{G}_\Gamma$. Alternatively, we place the label counterclockwise of the dart with respect to hyperedge on which it is 
incident. To distinguish between  vertices and hyperedges of $\Gamma$, 
elements in $\mathsf{E}(\Gamma)$ are shown as squares
whereas elements in $\mathsf{V}(\Gamma)$ are shown as circles, see Fig.~\ref{fig:dart-labeling}. 
If the hypergraph is connected then the bipartite graph is also connected and it is possible to associate a pair of 
permutations $\sigma, \tau \in S_n$ acting on the darts of the hypermap such that the group $\langle \sigma, \tau\rangle $  is transitive on the set of darts. Here $S_n$ is the symmetric group of permutations on $n$
elements. 

Let us consider a simple example. Consider a hypergraph with 2 vertices and 2 hyperedges, where
$\mathsf{V} = \{v_1,v_2 \}$ and 
$\mathsf{E} = \left\{(v_1,v_2,v_1,v_2 ), (v_1,v_2,v_1,v_2 )\right\} =\{e_1,e_2\}$.
The associated bipartite graph, see Fig.~\ref{fig:dart-labeling}, has $4$ vertices and $8$ darts. Vertices  can be repeated in a hyperedge. 
An embedding of this hypergraph, i.e. the embedding of its bipartite graph representation, on a torus has 4 faces, 4 vertices, 8 darts. 
\begin{center}
\begin{figure}[h]
\includegraphics{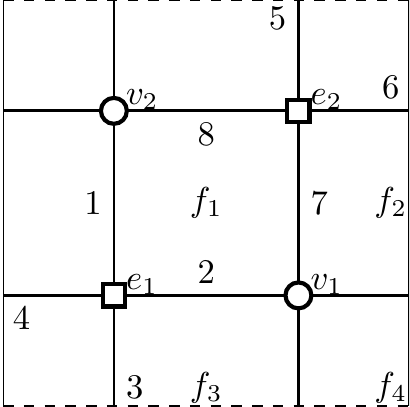}
\caption{A hypergraph embedded on a torus; opposite ends are to be identified. The circles are the vertices of the hypergraph. 
The square vertices are the hyperedges of the hypergraph. The labels are to the left as one moves from squares to the circles.}
\label{fig:dart-labeling}
\end{figure}
\end{center}
In Fig.~\ref{fig:dart-labeling} the permutations  $\sigma $  and $\tau$ are defined as follows:
\begin{eqnarray*}
\sigma &= & (\begin{array}{cccc}1&8&3&6\end{array})(\begin{array}{cccc}2&5&4&7\end{array})\\
\tau& = &(\begin{array}{cccc}1&2&3&4\end{array})(\begin{array}{cccc}5&6&7&8\end{array})\\
\sigma\tau^{-1}&=&(\begin{array}{cc}1&7\end{array})(\begin{array}{cc}2&8\end{array})(\begin{array}{cc}3&5\end{array})(\begin{array}{cc}4&6\end{array}),
\end{eqnarray*}
where $\sigma\tau^{-1}(i) = \sigma(\tau^{-1}(i))$.

The number of orbits of $\sigma $ is the number of vertices of original hypergraph while the number of orbits of $\tau $ is the number of hyperedges of the hypergraph. The number of orbits of $\sigma\tau^{-1}$ is the number of faces in the embedding of the hypergraph.
A dart belongs to a face if and only if the label is in the interior of the face. 
We use the following shorthand for simplicity:
\begin{eqnarray}
 \mathsf{F} &=& \mathsf{F}(\mathsf{G}_\Gamma) =\{ f_1,f_2,\ldots,f_{|\mathsf{F}|}\}, \\
 \mathsf{W} &=&\mathsf{E}(\mathsf{G}_\Gamma)=\{ w_1,w_2,\ldots,w_{|\mathsf{W}|}\}=\mathsf{W}(\Gamma), \\
 \mathsf{E} &=&\mathsf{E}(\Gamma)=\{ e_1,e_2,\ldots,e_{|\mathsf{E}|}\}\\
 \mathsf{V}&= &\mathsf{V}(\Gamma)=\{ v_1,v_2,\ldots,v_{|\mathsf{V}|}\}. 
 \end{eqnarray}
Denote the  binary vector spaces formed by taking $ \mathsf{F}(\mathsf{G}_\Gamma)$, $\mathsf{E}(\mathsf{G}_\Gamma)$, $\mathsf{E}(\Gamma)$ and $\mathsf{V}(\Gamma)$ as bases by 
$\mc{F}$, $\mc{W}$, $\mc{E}$, and $\mc{V}$ respectively. Topological codes are usually defined with respect to a series of boundary maps. We now proceed to define similar maps for the hypermaps with respect to the spaces just introduced. Define a ``boundary'' map for each face, dart, and  hyperedge as follows:
\begin{eqnarray}
d_2(f) & = & \sum_{i\in f}w_i \label{eq:d2}\\
d_1(w_i) &= &v_{\ni i}+v_{\ni \tau^{-1}(i)},\label{eq:d1}\\
\iota(e) &=& \sum_{i\in \delta(e)} w_i\label{eq:imap}
\end{eqnarray}
Recall that $v_{\ni i}$ is the unique vertex on which the dart $i$ is incident and $\delta(e)$ is set of darts incident on $e$. 
We can interpret the  map $d_1$ as giving the vertices of the hypergraph on which the half-edges $i$ and
$\tau^{-1}(i)$ are incident. 
We also define a projection map $p: \mc{W}\rightarrow \mc{W}/\iota(\mc{E})$ which enforces the following relations for each edge $e$:
\begin{eqnarray}
\sum_{i\in \delta(e)} w_i =0,\label{eq:edge-dep}
\end{eqnarray}
where $\delta(e)$ is the set of darts incident on $e$. So
\begin{eqnarray}
\mc{W}/\iota(\mc{E}) & =&\left \langle w_i  \bigg| w_i\in \mathsf{W}; \sum_{j\in \delta(e) } w_j =0 ; e\in \mathsf{E}(\Gamma)\right\rangle\\
p(w) &=& w +\iota(\mc{E})\label{eq:proj-W}
\end{eqnarray}
The hypermap-homology code is defined based on the following functions which lead to the chain complex in 
Eq.~\eqref{eq:chain-complex}:
\begin{eqnarray}
&&\partial_2 =p\circ d_2 \label{eq:doe2}\\
&&\partial_1(w+\iota(\mc{E})) = d_1(w)\label{eq:doe1}\\
&&\begin{CD}\label{eq:chain-complex}
\mc{F} @>\partial_2>> \mc{W}/\iota(\mc{E}) @>\partial_1>> \mc{V}\\
 \end{CD}.
\end{eqnarray}
In Fig.~\ref{fig:embedMaps} we summarize the various functions defined on the hypermap.
\begin{figure}[h]
\includegraphics{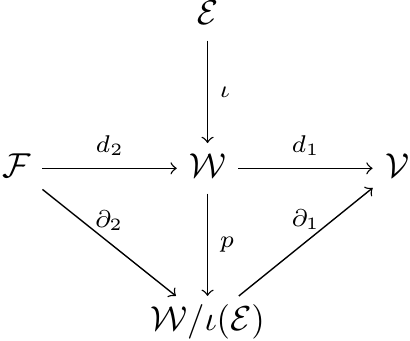}
\caption{Maps relating to the embedding of the hypergraph.}\label{fig:embedMaps}
\end{figure}

In order to define the quantum code we need to identify  a basis for  $\mc{W}/\iota(\mc{E})$. As there are different choices of bases, the maps $\partial_i$ could have different representations for $H_x$ and $H_z$. 
One canonical choice is as follows. Choose a set of darts $\mathsf{S}=\{s_1,s_2,\ldots, s_{m} \}$,  such that (i) $|\mathsf{S}|=|\mathsf{E}|$, and (ii) no two darts $s_i $ and $s_j$ in $\mathsf{S}$ are incident on the same hyperedge. 
These darts are called special darts. A basis for  $\mc{W}/\iota(\mc{E})$ is said to be special if 
it is  $\mathsf{W}\setminus \mathsf{S}$ and nonspecial otherwise. A quantum code of length 
$n=|\mathsf{W}|-|\mathsf{E}|$ is formed from the hypermap as follows.  The map $\partial_2$ can be represented by a $ (|\mathsf{W}|-|\mathsf{E}|)\times |\mathsf{F}|$ matrix $[\partial_2]=H_z^t$. 
The map $\partial_1$ can be represented by a $ |\mathsf{V}|\times (|\mathsf{W}|-|\mathsf{E}|)$ matrix 
$[\partial_1]=H_x$. The quantum code has the
stabilizer matrix 
\begin{eqnarray}
S=\left[\begin{array}{cc}H_x& 0 \\0 & H_z \end{array} \right]=\left[\begin{array}{cc}[\partial_1]& 0 \\0 & [\partial_2]^t \end{array} \right]. \label{eq:stabilizer-matrix}
\end{eqnarray}
For $S$ to define a stabilizer code we require $H_x H_z^t=0$; this is 
 ensured by the  condition $\partial_1 \circ \partial_2  = 0$, see
\cite[Proposition~4.12]{martin13} and the subsequent discussion  for proof.  A  hypermap code is said to be canonical if the basis is special and noncanonical otherwise.

We now illustrate the computations of the various maps and how they are used to construct a 
quantum code with reference to Fig.~\ref{fig:dart-labeling}. 
First identify a set of  special darts one for each hyperedge. Let us choose the set of special darts to be 
$\mathsf{S}=\{w_3,w_7 \}$. This set is not unique. In the present example they are darts below each hyperedge.
 The relations from the hyperedges are 
\begin{eqnarray}
w_1+w_2+w_3+w_4&=&0, \label{eq:equivReln-w3}\\
w_5+w_6+w_7+w_8&=&0. \label{eq:equivReln-w7}
\end{eqnarray} 
They can be used to eliminate the special darts in the computation of the images of $\partial_i$.
\begin{eqnarray}
\partial_2(f_1) &=& p(d_2(f_2))= p(w_2+w_8)=w_2+w_8\\
\partial_2(f_2) &=& p(w_1+w_7) = w_1+w_5+w_6+w_8 \label{eq:doe2f2}\\
\partial_2(f_3) &= &p(w_3+w_5) = w_1+w_2+w_4+w_5\\
\partial_2(f_4) &=& p(w_4+w_6)  = w_4+w_6
\end{eqnarray}
In Eq.~\eqref{eq:doe2f2} we can replace the special dart $w_7$ by linear combinations of nonspecial darts using Eq.~\eqref{eq:equivReln-w7} i.e.  $w_5+w_6+w_7+w_8=0$. Similarly for other faces. Note that $\partial_2(f_4) =\partial_2(f_1)+\partial_2(f_2)+\partial_3(f_3)$, so there are only $|\mathsf{F}|-1$ independent relations. 
 
We have $\mc{W}/\iota(\mc{E}) =  \langle {w}_1,{w}_2, \ldots, {w}_8|w_1+w_2+w_3+w_4=0,w_5+w_6+w_7+w_8=0\rangle$.
 So we can choose $B=\mathsf{W}\setminus \mathsf{S} = \{{w}_1,{w}_2, {w}_4, {w}_5, {w}_6, {w}_8\}$ as a basis for $\mc{W}/\iota(\mc{E})$. This basis is a special basis. The matrix representation of $\partial_2$ with respect to $B$ is $H_z^t$. 
 
  We then compute $\partial_1(w)$ for all $w\in B$.
For instance, $\partial_1(w_1) =v_{\ni 1} +v_{\ni\tau^{-1}(1)} = v_1+v_2$. Linearly, extending the action of $\partial_1$ we can perform   similar computations for any  $w\in \mc{W}/\iota(\mc{E})$. We obtain 
$H_x=\left[\begin{array}{cccccc} 1&1&1 & 1& 1&1 \\  1&1&1 & 1& 1&1 \end{array}\right]$.
 Thus the hypermap-homology code defined by Fig.~\ref{fig:dart-labeling} has the following stabilizer
 matrix. (We can remove the dependent rows of $H_x$ and $H_z$.)
\begin{eqnarray}
\left[\begin{array}{c|c}H_x& 0 \\\hline0 & H_z \end{array} \right] &=& 
\left[\begin{array}{c|c}\begin{array}{cccccc} 1&1&1 & 1& 1&1 \\1&1&1 & 1& 1&1 
\end{array} & 0 \\\hline0 & \begin{array}{cccccc} 
0&1&0 & 0& 0&1 \\
1&0&0 & 1& 1&1 \\
1&1&1 & 1& 0&0 \\
0&0&1 & 0& 1&0 
\end{array} \end{array} \right].\label{eq:hmap-code-stab}
\end{eqnarray}

If we choose an nonspecial basis say $B'=\{w_1, w_2'={w}_1 +{w}_2, {w}_4, {w}_5, {w}_6, {w}_8\}$. Then the matrices
$H_x$ and $H_z$ will be different and the associated quantum codes as well. In this case we compute 
\begin{eqnarray}
\partial_2(f_1) &= & w_2+w_8 = w_1+w_2'+w_8\\
\partial_2(f_2) &=& w_1+w_5+w_6+w_8 \\
\partial_2(f_3) &= &w_1+w_2+w_4+w_5 = w_2'+w_4+w_5\\
\partial_2(f_4) &= &w_4+w_6
\end{eqnarray}
Similarly, we can compute $\partial_1(w_2') = \partial_1(w_1+w_2) = v_{\ni 1}+v_{\ni \tau^{-1}(1)}+
v_{\ni 2}+v_{\ni \tau^{-1}(2)}=0$.
Thus the matrices $H_x$ and$H_z$  are given as 
\begin{eqnarray}
H_x = \left[ \begin{array}{cccccc}
1&0&1&1&1&1\\
1&0&1&1&1&1
\end{array}
\right]&;&
H_z = \left[ \begin{array}{cccccc}
1&1&1&0&0&0\\
1&0&0&1&1&1\\
0&1&1&1&0&0\\
0&0&1&0&1&0
\end{array}
\right].
\end{eqnarray}

Irrespective of the basis  chosen for $\mc{W}/\iota(\mc{E})$, the total number of encoded qubits is a function of the genus of the 
surface. In this sense we encode into the  topological degrees of freedom afforded by the surface.  
The distance on the other hand is not basis invariant. For instance, the hypermap code obtained from 
Fig.~\ref{fig:dart-labeling} with special basis has distance two while the code with nonspecial basis has 
distance one. The distance of hypermap codes was analyzed  when the basis was special \cite{martin13}. 
The distances of hypermap codes with nonspecial bases are a little more difficult to compute and it is not 
necessary they all have the same distance. For a special basis, the distance can be
related to the cycles of the hypermap and its dual, and thus given a topological interpretation, see for example \cite[Corollary~4.23]{martin13}. In case of a 
nonspecial basis such a topological interpretation of the distance has not yet been given. 

The following result  was proved  in  \cite{martin13} although it is not stated
as such therein. 
\begin{theorem}[Hypermap-homology codes\cite{martin13}]\label{th:hmap-codes}
Let $\Gamma$ be a hypergraph with $|\mathsf{E}|$ hyperedges and $\mathsf{G}_\Gamma$ its bipartite graph representation with $|\mathsf{W}|$ darts. The  hypermap obtained by embedding $\mathsf{G}_\Gamma$ on  a surface of genus $g$ leads to a  $[[|\mathsf{W}|-|\mathsf{E}|,2g]]$ quantum code. 
\end{theorem}
In Theorem~\ref{th:hmap-codes}, neither the choice of special darts nor the basis for $\mc{W}/\iota(\mc{E})$ 
is explicit,  but these choices must be made  before constructing the quantum code.  As was clear from the example, different choices of bases could lead to different codes with potentially different distances. 

\section{Hypermap codes and surface codes}\label{sec:main}
In this section, we address some of the questions raised in \cite{martin13}. We prove some new results about 
hypermap codes and then relate them to surface codes.   We show that the  hypermap codes with 
special and nonspecial bases are related by transformations involving just CNOT gates. 
Along the way we establish a correspondence between the qubits of the code and the  hypermap \footnote{
Recall that qubits are placed on the edges of the graph in case of surface codes. If we try to associate a 
qubit with  each dart we find that there are more darts than there are qubits. An additional problem is that 
$\mc{W}/\iota(\mc{E})$ has many bases in general. In such a situation each column of $H_x$ and $H_z$ correspond to linear combinations of the  labels of the darts. So at first sight it appears that we cannot have a direct correspondence between the  qubits and the darts. But fortunately this is not the case. 
}; this was only implicit in \cite{martin13}. Finally we relate  the hypermap codes to the surface codes showing an equivalence between  canonical hypermap codes and surface codes.

\subsection{Relation between hypermap codes of different bases}

Consider the bipartite representation of the hypergraph or its embedding. Let $\mathsf{S}$ be the collection of 
special darts. These special darts are $|\mathsf{E}|$ in number and no two of them are incident on the 
same hyperedge. Choose a special basis for $\mc{W}/\iota(\mc{E})$.
A special  basis for $\mc{W}/\iota(\mc{E})$ is of the form 
\begin{eqnarray}
\{w_{i_1},w_{i_2},\ldots, w_{i_n} \} =\mathsf{W}\setminus \mathsf{S} \label{eq:spl-basis}
\end{eqnarray}
where  $n=|\mathsf{W}|-|\mathsf{E}|$; $n$ is also the length of the code.  
Then we can place $n$ qubits, one on each of the nonspecial darts labeled $i_j\in \{i_1,i_2,\ldots, i_n \}$.  
The special darts carry no qubits. (Note that there are more darts than there are qubits. So the labels of the qubits  need  not be the same as labels of the darts on which the qubits are placed. It is possible to relabel the hypermap so that both qubit and dart labels coincide.)
Now define the stabilizer generators using the matrices $H_x$ and $H_z$. The matrix $H_z^t$ is 
simply the face-dart incidence matrix of the hypermap modulo $\iota(\mc{E})$ i.e. relations of the form given 
in  Eq.~\eqref{eq:edge-dep}. In other words, it is constrained to have no special darts. We can view  matrix $H_x$ as 
the vertex incidence matrix of nonspecial darts of the hypermap but the incidence vector is 
found after extending the half-edge $i$ to a full edge by combining $i$ and $\tau^{-1}(i)$. 

We would like to give a similar correspondence for the noncanonical codes and make precise the connection
between canonical and noncanonical hypermap codes. But first we need the
following lemma. 

\begin{lemma}\label{lm:rank-1up-factors}
Let $T$ be an invertible $n\times n $ binary matrix. Then $T$ and $T^{-1}$ can be decomposed as 
\begin{eqnarray}
T&=&R_{j_1}^{i_1} R_{j_2}^{i_2} \cdots R_{j_m}^{i_m},\label{eq:T-factors}\\
T^{-1}&=& R_{j_m}^{i_m}\cdots  R_{j_2}^{i_2}  R_{j_1}^{i_1},\label{eq:Tinv-factors}
\end{eqnarray}
where $R_{j}^{i} =I +e_{i} e_{j}^t$ and $m\leq n^2$.
\end{lemma}
\begin{proof}
Multiplying $T$ by $R_j^i$  from right adds the $i$th column to the 
$j$th column of $T$. Denote by $(T)_{i,j}$, the entry in $i$th row and $j$the column of $T$.
Suppose that $(T)_{i,i}=1$, then we can eliminate the nonzero entries $(T)_{i,j}$ in the $i$th row, for $1\leq j\neq i \leq n$ by adding the $i$th column to the $j$th column. If $(T)_{i,i}\neq 1$, then we can find some column $j$ such that $(T)_{i,j}=1$. Such a column must exist because $T$ is full rank. We can first add this column to $i$th column before eliminating the remaining nonzero entries $(T)_{i,j}$. 
If an entry $(T)_{i,j}$ is already zero we do not need to multiply by $R_j^i$. Assume now that 
we have made all entries in the $i$th row zero except the entry $(T)_{i,i}$. Since $T$ is full rank, the 
$(i+1)$st row is not identical to the $i$th row. So we can assume that there is a nonzero entry $(T)_{i+1,j}$ in some column $j\neq i$. Let us eliminate all the non-diagonal entries in $(i+1)$st row, i.e.  we eliminate all the nonzero entries except $(T)_{i+1, i+1}$. In this process  the $i$th row will not be affected because all its non-diagonal entries are zero. Starting from the first row we can reduce $T$ to the identity matrix by 
multiplying by 
matrices of the form $R_j^i$. The product of these matrices must 
equal $T^{-1}$. 
So we have 
\begin{eqnarray*}
T^{-1}&=& \prod_{k=1}^m R_{j_k}^{i_k} \\
T &= & \prod_{k=m}^1 (R_{j_k}^{i_k})^{-1} =  \prod_{k=m}^1R_{j_k}^{i_k},
\end{eqnarray*}
where the last equality follows from the fact that $R_{j}^{i}$ is its own inverse when $i\neq j$; $R_j^i R_j^i = I+e_ie_j^t+e_ie_j^t+e_ie_j^te_ie_j^t= I$.
Relabeling the $i_k$ and $j_k$, we obtain Eqs.~\eqref{eq:T-factors}~and~\eqref{eq:Tinv-factors}.

Each row $i$ requires no more than $n-1$ entries to be made nonzero; accounting for one additional multiplication when $(T)_{i,i}=0$, we need at most $n$ multiplications per row. Thus $m\leq n^2$.
\end{proof}

\begin{theorem} \label{th:hmap-code-equiv}
Suppose that an $[[n,k]]$  canonical hypermap code has  basis $B=\mathsf{W}\setminus \mathsf{S}$ and a noncanonical code has basis $B'=TB$.  The canonical code can be transformed to the noncanonical code by the application of $m \leq n^2$ CNOT gates where the $l$th CNOT gate is applied from
qubit $i_l$ to $j_l$. The number and location of CNOT gates is determined by the decomposition $T=R_{j_1}^{i_1} R_{j_2}^{i_2} \cdots R_{j_m}^{i_m}$, where $R_{j}^{i} =I +e_{i} e_{j}^t$.
\end{theorem}
\begin{proof}
The special basis $B$ for the canonical hypermap code can be given as in Eq.~\eqref{eq:spl-basis}. Then 
we can write the nonspecial basis $B'$ for the noncanonical hypermap code as 
\begin{eqnarray}
B'=\{ T w_{i_1},  T w_{i_2},\ldots, T w_{i_n} \} =T (\mathsf{W}\setminus \mathsf{S}), \label{eq:non-spl-basis}
\end{eqnarray}
where $T$ is an invertible (binary) matrix that transforms the special basis to the nonspecial basis. 
Now the columns of $H_x$ and $H_z$ are indexed by the basis vectors of $\mc{W}/\iota(\mc{E})$.
The relations between the representations of $\partial_i$ for different bases are as follows:
\begin{eqnarray}
[\partial_2]_{B'}=T^{-1}[\partial_2]_{B}  & \text{ and }& [\partial_1]_{B'} = [\partial_1]_{B} T. \label{eq:basis-change}
\end{eqnarray}
This ensures that $H_x$ and $H_z$ are orthogonal because $[\partial_1]_BTT^{-1}[\partial_2]_B =0$.

By Lemma~\ref{lm:rank-1up-factors}, we see that 
\begin{eqnarray}
[\partial_1]_{B'} &=& [\partial_1]_B T = [\partial_1]_B  R_{j_1}^{i_1} R_{j_2}^{i_2} \cdots R_{j_m}^{i_m}\label{eq:p1}\\
{[\partial_2]_{B'} }&=&   [\partial_2]_B(T^{-1})^t = [\partial_2]_B  (R_{j_m}^{i_m}  \cdots R_{j_2}^{i_2} R_{j_1}^{i_1})^t \nonumber\\
&=&[\partial_2]_B (R_{j_1}^{i_1})^t (R_{j_2}^{i_2})^t \cdots (R_{j_m}^{i_m})^t\nonumber\\
&=&[\partial_2]_B R_{i_1}^{j_1} R_{i_2}^{j_2} \cdots R_{i_m}^{j_m}\label{eq:p2}.
\end{eqnarray}
Let us parse  Eqs.~\eqref{eq:p1}~and~\eqref{eq:p2} equations closely. The first says that we add the column $i_l$ to the column
$j_l$ for $[\partial_1]_B$ while we add the column $j_l$ to the column  $i_l $ for $[\partial_2]_B$. This is precisely the action of a CNOT gate acting on qubits $i_l$ and $j_l$ with $i_l$ as control qubit, see for instance \cite[Lemma~2]{grassl03}. This proves that the transformation from the
canonical hypermap code to a noncanonical hypermap code can be achieved by the application of a sequence of CNOT gates
given by the decomposition of $T$.
\end{proof}

We make a few observations regarding the relation between canonical and noncanonical codes.
First the transformation assumes that the codes have been defined with respect to the same set of special darts. Different set of special darts could lead to different codes. With regard to the parameters of the canonical code we can  relate the dimension and distance to topological properties of the surface on which the hypergraph is embedded. For a noncanonical code while the dimension is related to the genus of the surface, the distance does not appear to have such a straightforward relation in general.  More importantly, 
the distance of the canonical code and the noncanonical code need not be the same. The transformation in Theorem~\ref{th:hmap-code-equiv} may not preserve distance. 
It is also possible that the stabilizer generators of the noncanonical code are not local, because there is no restriction on the nonspecial basis.  It is obvious that a noncanonical code can be converted to a canonical code by  application of  CNOT gates in the reverse order.  

If $T$ is simply a permutation matrix, then the code remains canonical. Since $T$ is composed of transformations of the form $R_{i_2}^{i_1}$, it is instructive to study the effect of one such transformation on the canonical code. Assume that we are applying $R_{i_2}^{i_1}$, in other words we are applying a CNOT between qubits $i_1$ and $i_2$, with $i_1$ as the control qubit. 
(We assume without loss of generality that the nonspecial darts have the
labels 1 to $n$ so that we have the same labels for qubits and nonspecial darts.) Suppose 
the qubits $i_1$ and $i_2$ are such that they are on adjacent darts in the hypermap as shown in Fig.~\ref{fig:hmap-face}, then the effect of CNOT on the hypermap is shown in Fig.~\ref{fig:hmap-cnot-adj-edges}.
The boundary of the face to which $i_1$ and $i_2$ belong is modified so that it no longer contains $i_1$.
The resulting code is still canonical with respect to the modified hypermap. 
\begin{center}
\begin{figure}[h]
\includegraphics{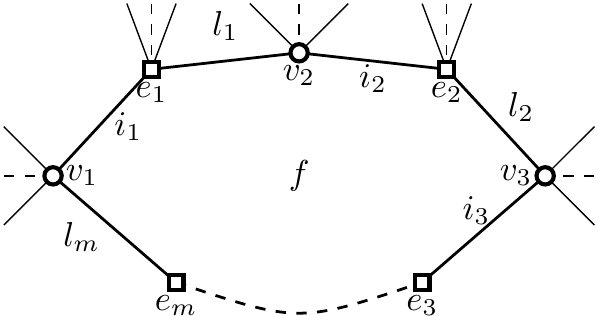}
\caption{A typical face in the embedding of the hypergraph. It has an even number of darts and exactly $|f|/2$ darts have their labels inside the face. The darts $l_\alpha$ and $i_\alpha$ are related as  $i_\alpha = \tau(l_\alpha)$.}
\label{fig:hmap-face}
\end{figure}
\end{center}

\begin{center}
\begin{figure}[h]
\includegraphics{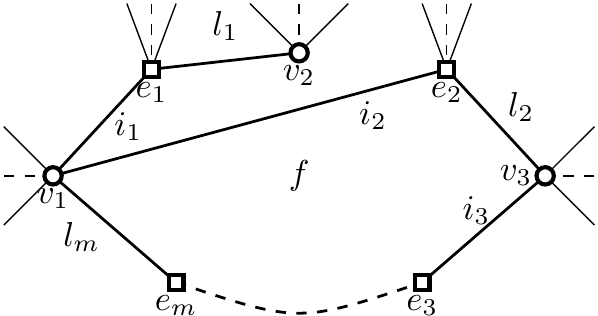}
\caption{The code with the basis $R_{i_2}^{i_1} B$ can be obtained by applying a CNOT gate on qubits $i_1$
and $i_2$ with $i_2$ as the target qubit.. For the special case when these qubits are adjacent, the effect can be understood graphically.  The resulting code is still canonical with respect to a slightly modified hypermap. A CNOT gate between non-adjacent qubits leads to a nontrivial modification of the hypermap in general.}
\label{fig:hmap-cnot-adj-edges}
\end{figure}
\end{center}
If the qubits are not adjacent, then $R_{i_2}^{i_1} B$ does not seem to effect such a simple modification to the hypermap. Understanding this transformation and relating it to the hypermap would be a very useful contribution in this context. 
 A combination of the transformations of $R_j^i$
could lead to a simple modification though. For instance swapping qubits $i_1$ and $i_2$ would lead 
to a relabeling of the hypermap and we would still end up with a canonical code. It would be interesting to find out which transformations $T$ can be described in terms of simple operations on the hypermap.

\subsection{Relation between hypermap codes and surface codes}

Our central result is that a canonical hypermap-homology code can be reduced to a surface code. The graph associated to this surface code can be derived from the hypergraph. Our approach is constructive 
and Algorithm~\ref{proc:hyper2surf} gives this transformation. We only consider hypergraphs 
whose bipartite representations are connected. 
If we are dealing with noncanonical hypermap code, then it can be associated to a canonical hypermap code. 
From Theorem~\ref{th:hmap-code-equiv} these two codes are related by a transformation involving CNOT gates alone.

\begin{theorem}\label{th:hmap-surf-equiv}
Every canonical hypermap-homological code is equivalent to a surface code. Given a hypergraph  $\Gamma$ (embedded on a surface),  Algorithm~\ref{proc:hyper2surf}  outputs a standard graph $\overline{\mathsf{G}}_\Gamma$ (embedded on the same surface) such that the surface code associated to $\overline{\mathsf{G}}_\Gamma$
 has the same stabilizer as the hypermap-homology code associated to $\Gamma$.
\end{theorem}

\renewcommand{\algorithmicrequire}{\textbf{Input:}}
\renewcommand{\algorithmicensure}{\textbf{Output:}}
\begin{algorithm}[H]
\caption{{\ensuremath{\mbox{ Surface code from a canonical hypermap  code}}}}\label{proc:hyper2surf}
\begin{algorithmic}[1]
\REQUIRE {A hypergraph $\Gamma$, and a set of special darts $\mathsf{S} \subset \mathsf{W}(\Gamma)$, such that 
$|\mathsf{S}|=|\mathsf{E}(\Gamma)|$ and every  dart in $\mathsf{S}$ is incident on a distinct hyperedge.}
\ENSURE {A graph $\overline{\mathsf{G}}_\Gamma$ such that the surface code of $\overline{\mathsf{G}}_\Gamma$   is same as the canonical hypermap-homology code.}
\STATE Form $\mathsf{G}$, the bipartite representation of the hypergraph. Note that $\mathsf{E}(\mathsf{G}_\Gamma) =\mathsf{W}(\Gamma)$.
\STATE For each hyperedge $e$ of the hypergraph choose one special dart ${s_e\in \mathsf{S}}$ such that $s_e\in \delta(e)$.
\STATE Embed $\mathsf{G}$ on a suitable surface; let the genus of the surface be $g$.
\STATE In each face $f$ of the embedding, draw new edges connecting vertex nodes of the hypermap. In other words, for each dart $i\in f$ join $v_{\ni i}$ to $v_{\ni \tau^{-1}(i)}$. Denote this graph by $\mathsf{G}'_\Gamma$.
\STATE Each face of $\mathsf{G}_\Gamma$ now contains  $|f|/2$ triangles; each triangle consists of two darts and one newly added edge. Exactly one label is present in each triangle. Label the new edge by that label. 
\STATE Modify  $\mathsf{G}'_\Gamma$ by deleting all the darts and the vertices associated to each hyperedge. Denote this graph by 
$\mathsf{G}''_\Gamma$
\STATE For each hyperedge $e$ there is a special dart ${s_e}$. Delete  the  edge in $\mathsf{G}''_\Gamma$ which has this label.
Denote the resulting graph as $\overline{\mathsf{G}}_\Gamma$.
\STATE The surface code defined by embedding of $\overline{\mathsf{G}}_\Gamma$ gives the same quantum code as the hypermap. The  stabilizer of the surface code is defined using Eqs.~\eqref{eq:vertex-op}--\eqref{eq:surf-stabilizer}.

\end{algorithmic}
\end{algorithm}

\begin{proof}
We show that Algorithm~\ref{proc:hyper2surf} leads to the same stabilizer code as the hypermap-homology code. More precisely, we show that  that the matrices $H_x$ and $H_z$ arising in the hypermap-homology construction are exactly the vertex-edge and face-edge incidence matrices of the graph $\overline{\mathsf{G}}_\Gamma$. 

Consider first the stabilizer of the hypermap-homolgy code. Each face leads to a $Z$-only stabilizer generator.  Note that  ${H}_z$ is a $(|\mathsf{W}|-|\mathsf{E}|)\times |\mathsf{F}|$ matrix.
Although there are only $|\mathsf{F}|-1$ independent generators though, see \cite{martin13}, we include the generator from each face in $H_z$.
Let $f$ be a face of the hypermap.  The labeling of the darts in the hypermap is such that only half the darts
that constitute the boundary of $f$ have their labels inside $f$. This is illustrated in Fig.~\ref{fig:hmap-face}.

We can choose the set of nonspecial darts as a basis for $\mc{W}/\iota(\mc{E})$.  If $w_i$ is not special then 
we can write $p(w_i)=w_i$ otherwise  we can write $p(w_i)= \sum_{j\neq i\in \delta (e_{\ni i})} w_j$.
The latter follows from the relation $\sum_{j\in \delta(e)} w_j=0 $ for every  hyperedge $e$.
Then we obtain 
\begin{eqnarray}
\partial_2(f) &=& p (d_2(f)) = p\left(\sum_{i\in f} w_i\right)\label{eq:f-bndry-hmap}\\
& =&\sum_{\stackrel{i\in f}{i\not\in \mathsf{S}} } p(w_i)+\sum_{\stackrel{i\in f}{i\in \mathsf{S}} } p(w_i)
=\sum_{\stackrel{i\in f}{i\not\in \mathsf{S}} }w_i+ \sum_{\stackrel{i\in f}{i\in \mathsf{S}} } \sum_{\stackrel{j\in \delta(e_{\ni i})}{j\neq i }} w_j\nonumber
\end{eqnarray}
The row associated to $\partial_2(f)$ in $H_z$ is simply the characteristic vector of $\partial_2(f)$.

Now let us consider the $Z$-type stabilizer generators of $\overline{\mathsf{G}}_\Gamma$. Let us begin with
the graph $\mathsf{G}'_\Gamma$. With respect the hypermap it has additional edges. The addition of edges between $v_{\ni i}$ and $v_{\ni \tau^{-1}(i)}$  transforms the face $f$ in Fig.~\ref{fig:hmap-face} as shown in Fig.~\ref{fig:hmap-face-newedges}. This creates $|f|/2$
new (triangular) faces within $f$, each of which is bounded by two darts of the hypermap and one new edge. Furthermore,
exactly one of the labels is contained in each of these new faces and every label is contained in some 
triangle. Therefore, we can label each   new edge by a unique label; furthermore note that these new edges 
are are exactly $|\mathsf{W}|$ in number. The face $f$ is modified so that it only contains the new edges 
and  the vertices of the hypergraph in its boundary. We  label this derived  face also by $f'$. The  derived face $f'$ has only the newly added edges in its boundary. 
In fact we have $\partial(f')=\sum_{i\in f'} w_i$, where the boundary is with respect to  $\mathsf{G}'_\Gamma$. The boundary of $f'$ in $\mathsf{G}'_{\Gamma}$ is same as the boundary of $f$ in the hypermap.

\begin{center}
\begin{figure}[h]
\includegraphics{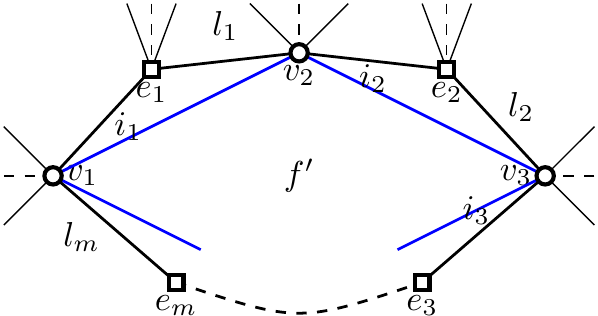}
\caption{(Color online) Illustrating the addition of edges between $v_{\ni i}$ and $v_{\ni \tau^{-1}(i)}$. The additional edges are shown in color. We call them edges to distinguish them from the darts (half-edges) of the hypermap. Observe that $\partial(f')=\partial(f)$.}
\label{fig:hmap-face-newedges}
\end{figure}
\end{center}

Let us consider the transformation of the hyperedges due to the addition of the new edges. Because of the 
edges added between the vertices $v_{\ni i}$ and $v_{\ni \tau^{-1}(i)}$, exactly one label is adjacent to any 
new edge. So we can label the newly added edges by the label in the triangle, this is illustrated in 
Fig.~\ref{fig:hmap-hedge-transformation-1}. 
\begin{center}
\begin{figure}[h]
\includegraphics{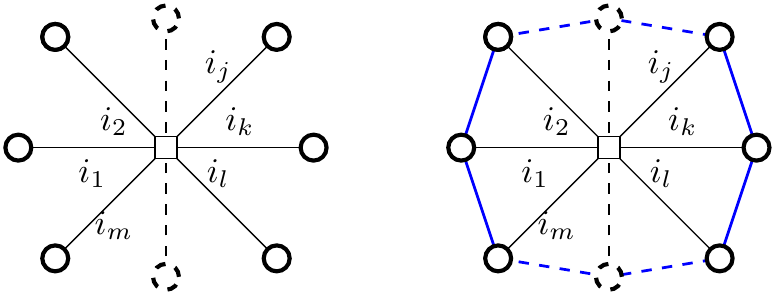}
\caption{A hyperedge in the embedding of the hypergraph $\Gamma$. The addition of edges leads to a creation of triangular faces each containing exactly one label. }
\label{fig:hmap-hedge-transformation-1}
\end{figure}
\end{center}
The deletion of hyperedges and the darts incident on them leaves each face $f'$ unchanged. So $f'$ is also a 
face of $\mathsf{G}''_\Gamma$. So we can denote without ambiguity the face derived from $f'$ as $f''$.
Furthermore, deletion of the hyperedges and the
darts transforms each hyperedge  to a face in $\mathsf{G}''_{\Gamma}$, see
Fig.~\ref{fig:hmap-hedge-transformation-2}. Thus exactly 
$|\mathsf{E}|$ new faces are added to $\mathsf{G}''_\Gamma$ with respect to the hypermap; and they can be indexed by the hyperedges.
\begin{center}
\begin{figure}[h]
\includegraphics{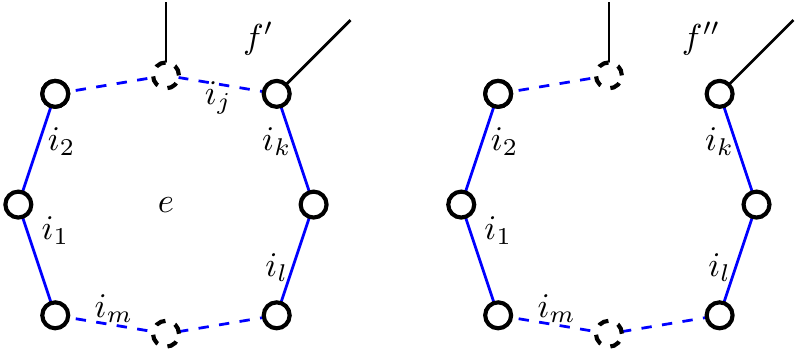}
\caption{The deletion of the darts and the hyperedge creates a new face with exactly one special edge. On the deletion of the special edge, say $i_j$, it gets merged with exactly one face $f$. The boundary of $f$ now includes the rest of the edges of $e$ i.e. excepting the special edge $i_j$.}
\label{fig:hmap-hedge-transformation-2}
\end{figure}
\end{center}
Consider any such face $e$. Only one of the edges in the boundary of $e$  has the same label as some special dart. We call such edges special edges. 
Now when such a special edge in the boundary of $e$ is deleted, then $e$ is merged with exactly one face of 
$f''$ of $\mathsf{G}''_\Gamma$, since an edge is shared only by two faces, see 
Fig.~\ref{fig:hmap-hedge-transformation-2}. While $f''$ can be merged with many faces $e_i$ derived from the hyperedges,  it is not 
merged with any other face derived from the faces of the hypermap,  because such derived faces do not share 
any edges.
Thus $\overline{\mathsf{G}}_\Gamma$ has exactly as many faces as the hypermap and the face derived from the
merging of $f''$ and $e$ (and possibly other faces which share a special edge with $f''$) can be labeled uniquely as $\overline{f}$.    
 Let us look at the boundary  $\overline{f}$ in $\overline{\mathsf{G}}_\Gamma$. 
The boundary of $f''$ in $\mathsf{G}''_\Gamma$  is the same as the boundary of $f$ in $\mathsf{G}_\Gamma$ and is equal to $\sum_{i\in f} w_i$. When all the special edges are deleted, $f''$ may be merged with other faces  $e\in \mathsf{F} (\mathsf{G}''_\Gamma)$ which share a special edge with $f''$. The 
 resulting face $\overline{f}$ has a boundary 
that is given by union of their boundaries, (the sum is take modulo 2). The boundary of $f''$ is 
$\sum_{i\in f} w_i$. Each special edge $i_j$ that is removed causes the boundary to include the remaining edges bounding $e_{\ni i_j}$. The boundary of $e$ is $\sum_{i\in \delta(e)}w_i$.
Thus the boundary of $\overline{f}$ in $\overline{\mathsf{G}}_\Gamma$ is 
$\sum_{\stackrel{i\in f}{i\not\in S}} w_i + \sum_{\stackrel{i\in f}{i\in S}} \sum_{\stackrel{j\in \delta(e_{\ni i})}{j\neq i}}w_j$, which is precisely $\partial_2(f)$, the boundary of $f$ in the hypermap, as 
computed in  Eq.~\eqref{eq:f-bndry-hmap}.
Thus the stabilizer generator associated with  a face in $\overline{\mathsf{G}}_\Gamma$ is same as the 
generator associated to its parent face in the hypermap.  This shows that $H_z$ is the same as the 
face-edge incidence matrix of  $\mathsf{G}_\Gamma$.

It remains now to show that the matrix $H_x$ is the same as the vertex-edge incidence matrix of 
$\overline{\mathsf{G}}_\Gamma$. The matrix $H_x$ is determined by the map $\partial_1$ and
acts on the space  $\mc{W}/\iota(\mc{E})$. As mentioned earlier,  we can take the nonspecial darts as a 
basis for $\mc{W}/\iota(\mc{E})$.  Then the columns of $H_x$ are given by  the characteristic vector of 
$\partial_1(w_i)$, where $i\not\in S$.  We have $\partial_1(w_i) = v_{\ni i }+v_{\ni \tau^{-1}(i)}$. But 
the vertices $v_{\ni i}$ and $v_{\ni \tau^{-1}(i)}$  are connected by an edge in 
$\overline{\mathsf{G}}_\Gamma$. This edge is also labeled $i$  in $\overline{\mathsf{G}}_\Gamma$. 
Hence, we can obtain the $i$th column of 
$H_x$ by considering the incidence vector of every edge in $\overline{\mathsf{G}}_\Gamma$. The columns 
put together then give the vertex-edge incidence matrix of $\overline{\mathsf{G}}_\Gamma$.
Thus the surface code generated by embedding of $\overline{\mathsf{G}}_\Gamma$ has the same stabilizer as 
the hypermap-homology code on $\Gamma$. This proves that every hypermap-homology code can be realized by an 
equivalent  surface code. 
\end{proof}

Although Theorem~\ref{th:hmap-surf-equiv} does not mention, the choice of special darts is made explicit in 
Algorithm~\ref{proc:hyper2surf}. We give a simple example that illustrates the application of Theorem~\ref{th:hmap-surf-equiv}. Consider the hypermap given in Fig.~\ref{fig:dart-labeling}.  Draw additional edges in each face connecting two adjacent circles
as one goes around the face. The hypermap in Fig.~\ref{fig:dart-labeling} is modified as show in Fig.~\ref{fig:hmap-code-2-surf-code-1}.

\begin{center}
\begin{figure}[h]
\includegraphics[scale=0.9]{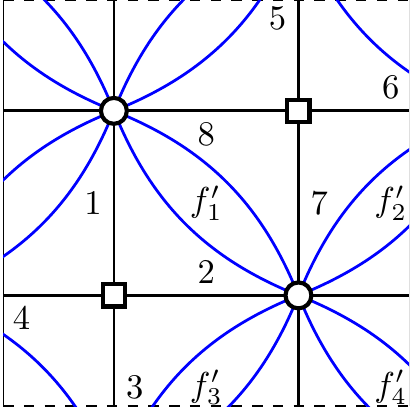}
\caption{(Color online) The graph $\mathsf{G}'_{\Gamma}$ for the hypermap in Fig.~\ref{fig:dart-labeling}.
It is obtained by the addition of new edges to the hypermap. Any two adjacent darts and a newly added edge form a triangle which contains exactly one label; therefore the newly added edge can be uniquely identified by a label.}
\label{fig:hmap-code-2-surf-code-1}
\end{figure}
\end{center}

Next we modify the graph in Fig.~\ref{fig:hmap-code-2-surf-code-1} as follows. We remove all the original darts and edges of the hypergraph. 
In addition we also remove the special darts. These transformations are shown in Figs.~\ref{fig:hmap-code-2-surf-code-2}~and~\ref{fig:cube-surf-code} respectively.
\begin{figure}[h]
 \centering
 \subfigure[]{
 \includegraphics[scale=.9]{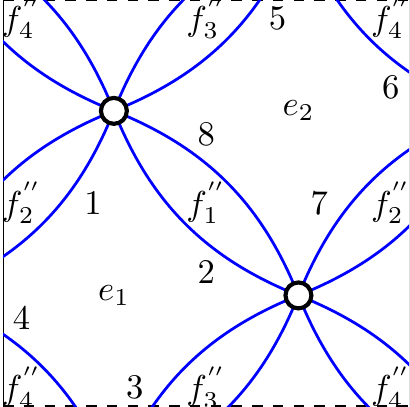}
    \label{fig:hmap-code-2-surf-code-2}
}
\subfigure[]{

 \includegraphics[scale=.9]{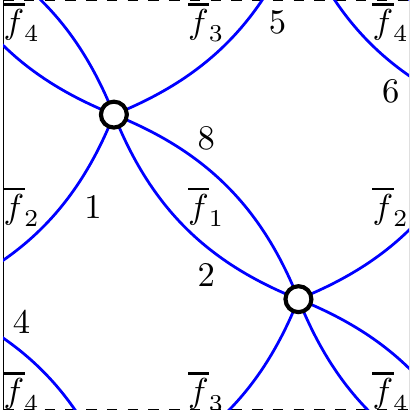}
   \label{fig:cube-surf-code}
   }
 \caption[Transforming the hypermap-homology code to a surface code.]{
(Color online) (a) The graph $\mathsf{G}''_{\Gamma}$ obtained by the removal of the darts and hyperedge-vertices in $\mathsf{G}'_{\Gamma}$; the hyperedges are transformed to faces in $\mathsf{G}''_{\Gamma}$. (b) The graph 
$\overline{\mathsf{G}}_\Gamma$ obtained by removing the special edges i.e. \{3,7 \}; each hyperedge-face $e_i$ is merged with exactly one face. $\overline{\mathsf{G}}_\Gamma$ has the same stabilizer as the hypermap-homology code of Fig.~\ref{fig:dart-labeling}.}
 \label{fig:insertRank3}
\end{figure}	
The hypermap-homology code obtained from Fig.~\ref{fig:dart-labeling} is  identical to the surface code obtained from Fig.~\ref{fig:cube-surf-code}. Consider $\mathsf{I}_\mathsf{V}$ and $\mathsf{I}_\mathsf{F}$, the vertex-edge and face-edge incidence matrices of $\overline{\mathsf{G}}_\Gamma$ in Fig.~\ref{fig:cube-surf-code}:
\begin{eqnarray}
\mathsf{I}_\mathsf{V}=\left[\begin{array}{cccccc} 1 & 1 & 1 & 1 & 1&1\\
1 & 1 & 1 & 1 & 1&1 \end{array} \right] &\quad&
\mathsf{I}_{\mathsf{F}}=\left[\begin{array}{cccccc} 
0&1&0 & 0& 0&1 \\
1&0&0 & 1& 1&1 \\
1&1&1 & 1& 0&0 \\
0&0&1 & 0& 1&0 
\end{array} \right]\label{eq:surf-code-stab}
\end{eqnarray}
The stabilizer matrix of the surface code is $\left[\begin{array}{cc} \mathsf{I}_\mathsf{V}&0 \\0 & \mathsf{I}_\mathsf{F} \end{array}\right]$.
We can see that this is the same as the stabilizer matrix of the hypermap-homology code given in Eq.~\eqref{eq:hmap-code-stab}.

The substance of Theorem~\ref{th:hmap-surf-equiv} is that the procedure illustrated works in general and 
reduces a canonical hypermap-homology code to a surface code. 
The converse of Theorem~\ref{th:hmap-surf-equiv}, i.e. every surface code is also a (canonical) hypermap-homology code, is straightforward. 

\begin{corollary}\label{co:canonical-hmap-surf-equiv}
Every canonical hypermap-homology code is equivalent to a surface code and vice versa. 
\end{corollary}
\begin{proof}
 We only sketch the proof. Every graph $\mathsf{G}$ can also be 
viewed as a hypergraph, denote it also by $\Gamma$. The bipartite graph representation of this hypergraph is obtained by  
taking the original graph and splitting every edge in $\mathsf{G}$ into two edges and adding a new vertex
for each edge.  Then we can proceed with the construction proposed in \cite{martin13} to obtain a 
hypermap-homology code. At this point we have two codes: a hypermap code and a surface both derived from 
the same graph $\mathsf{G}$. But it is by no means obvious that they are identical. We show that they are
the same code. Note that every hyperedge in the hypermap has only 
two darts incident on it. One of these can be chosen  as a special dart. Now if we apply 
Algorithm~\ref{proc:hyper2surf}, and trace through the various transformations, we find that every hyperedge is transformed to a face with two edges in 
$\mathsf{G}_\Gamma''$. This graph is identical to the graph  obtained from $\mathsf{G}$ where every edge is 
replaced by two edges. Hence, $\overline{\mathsf{G}}_\Gamma$ obtained after the deletion of all the special edges from $\mathsf{G}_\Gamma''$  would be same as the original graph $\mathsf{G}$. Thus every surface code is equivalent to a hypermap-homology code. This together with 
Theorem~\ref{th:hmap-surf-equiv} implies the corollary.
\end{proof}

Combining Corollary~\ref{co:canonical-hmap-surf-equiv} and Theorem~\ref{th:hmap-code-equiv} we obtain  
the following result:
\begin{corollary}\label{co:hmap-surf-equiv}
Any $[[n,k]]$ hypermap code is either identical to a surface code or it can be transformed  to a surface code with the application of  $m\leq n^2$ CNOT gates.
\end{corollary}
The CNOT gates are required only if the hypermap code is noncanonical. 
The sequence in which the CNOT gates are applied in the transformation in Corollary~\ref{co:hmap-surf-equiv} is reverse of the sequence in Theorem~\ref{th:hmap-code-equiv}.  In other words it is the sequence of CNOT gates required to transform the noncanonical hypermap code into a canonical code.  
While a noncanonical code may be transformed to a surface code, it is not necessary that it has the same 
parameters as the resulting surface code. In particular, the distance can change. Therefore some noncanonical 
codes may be realized only from the embedding of hypergraphs and not by the embedding of graphs. 


\section{Conclusion and Discussion}\label{sec:disc}
Our results imply that  canonical hypermap-homology codes cannot
improve upon the parameters of surface codes.  Noncanonical hypermap codes may have better parameters than 
surface codes. An interesting problem in this context is to construct noncanonical hypermap codes that have better parameters than  surface codes. In such a construction, we also want to be able to preserve the locality of the stabilizer generators. Understanding the distance and decoding of the noncanonical codes
merits further study. Hypermap codes provide a  different 
perspective on topological  codes which might yield new insights into their properties and potentially lead to better quantum codes and alternate decoding algorithms for topological codes.  

Hypermap-homology codes introduce in a very concrete fashion the use of standard hypermap homology into the study of quantum codes.  The use of hypermap homology  in the construction of quantum codes is an important development and  its applications are  yet to be fully explored in the context of quantum codes. 
Other  results such those in \cite{tillich09,kovalev12} suggest that hypergraphs offer a fertile ground 
for construction of new quantum codes. The use of hypergraphs has been fruitful in subsystem codes as well and  it would be  interesting to study homology of hypermaps in that context also \cite{suchara10}. 
Homology of hypergraphs is considered with a different perspective in \cite{suchara10,pra13}. It seems that the notion of homology used therein does not entirely coincide with that used in \cite{martin13}. It would be an interesting problem to relate the two.


\def\cprime{$'$}

\end{document}